\documentclass[a4paper,reqno,12pt]{amsart}
\usepackage{geometry}
\usepackage{courier}
\usepackage{times}
\usepackage{amsmath,amssymb}
\usepackage{enumitem}
\usepackage[foot]{amsaddr}
\usepackage[normalem]{ulem}
\usepackage{chicago}

\allowdisplaybreaks[1]
\theoremstyle{plain}
\newtheorem{proposition}{\bf Proposition}
\newtheorem{theorem}{\bf Theorem}
\newtheorem{lemma}{\bf Lemma}
\newtheorem{corollary}{\bf Corollary}

\theoremstyle{definition}
\newtheorem{definition}{\bf Definition}
\newtheorem{example}{\bf Example}
\newtheorem{remark}{\bf Remark}

\thanks{$^*$This research was prompted by Dan Felsenthal and Mosh\'{e} Machover's discussion on
``Voting power when voters' independence is not assumed'' at the \emph{Voting Power \& Procedures Workshop 2007}
in Warwick. The authors acknowledge helpful feedback from seminar audiences in Bayreuth, Bielefeld, Bilbao, Jena
and Turku.}

\begin{document}
\title[The Prediction value]{The Prediction value$^*$}
\author{Maurice Koster$^a$}
\address{\textnormal{$^A$Corresponding author, CeNDEF, Amsterdam School of
Economics, University of Amsterdam, e: mkoster@uva.nl, t: ++31 20 525 4226.}}
\author{Sascha Kurz$^b$}
\address{\textnormal{$^B$Department of Mathematics, University of Bayreuth.}}
\author{Ines Lindner$^c$}
\address{\textnormal{$^C$Department of Economics \& Econometrics, VU
University, Amsterdam.}}
\author{Stefan Napel$^d$}
\address{\textnormal{$^D$Department of Economics, University of Bayreuth;
Public Choice Research Centre, Turku.}}
\date{\today }
\maketitle


{\small \noindent \textbf{Abstract:} We introduce the {prediction value}
(PV) as a measure of players' informational importance in probabilistic TU
games. The latter combine a standard TU game and a probability distribution
over the set of coalitions. Player~$i$'s prediction value equals the
difference between the conditional expectations of $v(S)$ when $i$
cooperates or not. We characterize the prediction value as a special member
of the class of (extended) values which satisfy anonymity, linearity and a
consistency property. Every $n$-player binomial semivalue coincides with the
{PV} for a particular family of probability distributions over coalitions.
The {PV} can thus be regarded as a power index in specific cases.
Conversely, some semivalues~-- including the Banzhaf but not the Shapley
value~-- can be interpreted in terms of informational importance. }

{\small \medskip }


\bigskip 

\section{Introduction}

\noindent Concepts of power and importance in models of cooperation are
central to numerous studies in sociology, political science, mathematics,
and economics. Much of the literature applies values or power indices which
attribute fixed roles~-- often perfectly symmetric~-- to all players in the
underlying coalition formation process and then focus on their \emph{%
marginal contributions}. Most prominent examples are the \emph{Shapley value}
and \emph{Banzhaf value} (Shapley 1953; Banzhaf 1965); others can be
found in \citeN{Roth:1988}, \citeN{Owen:1995}, %
\citeN{Felsenthal/Machover:1998} or \citeN{Laruelle/Valenciano:2008}.

A player who makes a positive marginal contribution, i.e., who can raise
some coalitions' worth by joining in, or lowering it by leaving, is
considered as important and powerful. Others who never affect a coalition's
worth $v(S)$ are referred to as \emph{dummy} or \emph{null players}. The
powerful ones are attributed a positive share of the decision body's
aggregate ability to implement collective decisions or to create surplus;
the indicated value concepts differ just in how marginal contributions to
distinct coalitions are weighted. For instance, the Shapley value weights a
player~$i$'s marginal contribution to a coalition $S\not\ni i$ according to
the total number of ordered divisions of the reduced player set $N\setminus
i $ into members of $S$ and its complement; the Banzhaf value weights $i$'s
marginal contributions equally for all $S\subseteq N\setminus i$.\footnote{%
We adopt the usual notational simplifications like writing $S\setminus i$
or $S\cup ij$ instead of $S\setminus \{i\}$ or $S\cup\{i,j\}$.
}

With an appropriate rescaling, weights on specific marginal contributions
can be interpreted as a probability distribution. So Shapley value, Banzhaf
value, and more generally \emph{probabilistic values} \cite{Weber:1988} correspond
to the \emph{expectation of a difference}. This difference is
between the worth of a random coalition $S$ that is drawn from $%
2^{N\setminus i}$ according to a value-specific probability distribution $%
P_{i}$ and the worth of the same coalition when $i$ joins, i.e., a
probabilistic value equals $\mathbb{E}_{P_{i}}[v(S\cup i)-v(S)]$ for a fixed
family of distributions $\{P_{i}\}_{i\in N}$.\footnote{%
More precisely, a probabilistic value draws on a \emph{family of families}
of distributions, parameterized by the player set $N$. One may equivalently
consider suitable probability distributions $P_{i}$ on $\{S\in 2^{N}\colon
i\in S\}$ and then evaluate $\mathbb{E}_{P_{i}}[v(S)-v(S\setminus i)]$.}

The expectation of a difference, however, can behave in
strange ways when the family of distributions $\{P_{i}\}_{i\in N}$ implicate
correlated voting behavior. This can be the case, for example, when voting
is preceded by a process of information transmission or opinion formation.%
 \footnote{See, for example, the seminal opinion formation model of \citeN{DeGroot:1974}:
individuals start with initial opinions (beliefs) on a subject represented by an $n$%
-dimensional vector of probabilities, and repeatedly update their individual opinion based on the current opinions of their peers. Different structures of consensus formation can be captured by different network topologies.} The following example, which we owe to Mosh\'{e} Machover,
illustrates the conceptual problem. 

\begin{example}
\label{example1} Consider the canonical simple majority decision rule with
an assembly of $5$ voters. Let $P$ be the probability distribution that
assigns probability $0$ to the $20$ coalitions containing exactly two or
exactly three voters; and equal probability of $1/12$ to each of the
remaining $12$ divisions. Here, the probabilistic value $\mathbb{E}%
_{P_{i}}[v(S\cup i)-v(S)]$ is zero for all players. That \emph{no} member of this decision body should have any voting power or importance is somewhat counterintuitive however.
\end{example}

This paper proposes an alternative approach: namely, to consider the \emph{%
difference of two expectations}. These expectations will be derived from a
given probabilistic description $P$ of coalition formation. The latter plays
a similar role as $\{P_{i}\}_{i\in N}$ for probabilistic values or
corresponding families $\{P_{i}^{v}\}_{i\in N}$ for values that evaluate
marginal contributions in game $v$-specific ways.\footnote{%
This is, for instance, the case when positive probability is only attached
to \emph{minimal winning coalitions} (see, e.g. Holler 1982 and Holler and Li 1995).} However, we take
$P$ as a primitive of the collective decision situation under investigation,
rather than of the solution concept.

We thus depart from the literature in two respects: first, we consider
\emph{probabilistic games} $(N,v,P)$ where $(N,v)$ is a standard TU game and
$P$ is a probability distribution on $N$'s power set $2^N$. Second, we
introduce a new value that reflects the difference between two conditional
expected values. Specifically, we define the
\emph{prediction value\ (PV)} of any given player $i\in N$ as the difference
in $v$'s expected value when the distribution $P|i$ which conditions $P$ on
the event $\{i\in S\}$ and the distribution $P|\neg i$ which conditions on $%
\{i\notin S\}$ are applied. In other words, we suggest to evaluate $\mathbb{E%
}_{P|i}[v(S)] - \mathbb{E}_{P|\neg i}[v(S)]$ instead of $\mathbb{E}%
_{P_i}[v(S \cup i)-v(S)]$. The two coincide in interesting special cases,
but not in general.

The difference between the respective conditional expectations can be
interpreted as the
importance of a player in the probabilistic game $(N,v,P)$ in several ways.
Most generally, it captures the informational or predictive value of knowing
$i$'s decision in advance of the process which divides $N$ into some final
coalition $S$ and its complement. Moreover, in case $i$'s membership of the
coalition which supports a specific bill or cooperates in a joint venture is
statistically independent of others, the {PV} provides a measure of $i$'s
influence on the outcome of collective decision making, or of $i$'s power in
$(N,v,P)$.

A null player who, say, has a voting weight that cannot matter for matching
a required threshold \emph{and} whose behavior is uncorrelated with the
remaining players has a {PV} of zero. Endowing the same player with greater
voting weight will at some point translate into a positive value~--
reflecting the difference that her vote can now make for the outcome.
Leaving initial voting weights unchanged, the PV will also ascribe positive
importance to the null player if interdependencies make its cooperation a
predictor of whether a proposal is passed.

Plausible causes for dependencies abound and, for instance, include the
possibility that the player in question is actually without vote but
`followed' by the official voters (as, say, their paramount or supreme leader). The
proposed change of perspective~-- from, traditionally, the expected
difference that a player would make by an ad-hoc change of coalition
membership towards the difference in expectations for the collective outcome
which is associated with that player's cooperation~-- opens the route to
studying voting and coalition formation as the result of social interaction.
Final votes may be determined by whether $i$ is initially a supporter or
opponent even if $i$ is a null player of $(N,v)$, and this is arguably a
source of power just like official voting weight. We believe that evaluating
changes in conditional expectations can help to quantify this in future
research.

Here, we primarily want to introduce and investigate the {prediction value}.
We formally define it in Section~\ref{sec:definition}.
We describe a set of
characteristic properties in Section~\ref{sec:characterization} and relate
the {PV} to traditional probabilistic values in Section~\ref{sec:relation_to_prob_values}.
The considered distributions $P$ could embody the \emph{a~prioristic}
presumptions of traditional power measures, i.e., be the uniform
distribution on $2^N$ or the space of permutations on $N$. (Interestingly,
the latter does \emph{not} make {PV} and Shapley value coincide.) But $P$
could equally well be based on empirical data~-- say, observations of past
voting behavior in a decision making body like the US~Congress, EU~Council
of Ministers, etc. We briefly conduct such \emph{a~posteriori analysis} with
the PV in an application to the Dutch Parliament in Section~\ref{sec: DP}
and conclude in Section~\ref{sec:conclusion}.


\section{Probabilistic games and the {prediction value}}

\label{sec:definition} \noindent A \emph{TU game} is an ordered pair $(N,v)$
where $N\subset \mathbb{N}$ represents a non-empty, finite set of players
and $v\colon 2^N\rightarrow \mathbb{R} $ is the characteristic function
which specifies the worth $v(S)$ of any subset or coalition $S\subseteq N$
and satisfies $v(\varnothing)=0$. The set of all TU games is denoted by $%
\mathcal{G}$, and the set of all TU games with player set $N$ by $\mathcal{G}%
^N$. The cardinality of a finite set $S\subset \mathbb{N}$ is denoted $|S|$.

$(N,v)\in\mathcal{G}$ is a \emph{simple game} if $v$ is a monotone Boolean
function, i.e., $v(S)\le v(S^{\prime })$ for all $S\subseteq S^{\prime
}\subseteq N$, such that $v(\varnothing)=0$ and $v(N)=1$. Given any
non-empty coalition $S\subseteq N$, the so-called \emph{unanimity game} $u_S$
is defined by $u_S(T)=1$ if $S\subseteq T$ and $u_S(T)=0$ otherwise. Note
that we will drop the player set $N$ from our notation when it is clear from
the context; so $u_S$ is shorthand for $(N,u_S)$. Moreover, we refer to $%
u_{\{i\}}$ simply as $u_i$.

A \emph{probabilistic game} is an ordered triple $(N,v,P)$, where $(N,v)$ is
a TU game and $P$ is a probability distribution on the power set of $N$, $%
2^N $. The set of all probabilistic games is denoted by $\mathcal{PG}$; and $%
\mathcal{PG}^N$ is the restriction to the class of probabilistic games with
player set $N$.

A \emph{TU value} is a function which assigns a real number to all elements
of $N$ for any given TU game. An \emph{extended value} is a mapping $\varphi$
that assigns to each probabilistic game $(N,v,P)$ a vector $%
\varphi(N,v,P)\in \mathbb{R} ^{|N|}$. $\varphi_i(N,v,P)$ will be interpreted
as a measure of the `difference', in an abstract sense, that player~$i$
makes for the probabilistic game $(N,v,P)$. It might, for instance, relate
to the average of marginal contributions $v(S\cup i)-v(S)$ that are made by $%
i$ to coalitions $S\in N\setminus i$, to the difference that $i$ makes to a
potential function (i.e., a mapping from $\mathcal{PG}$ to $\mathbb{R} $)
when $i$ is added to the player set $N^{\prime }$ such that $N^{\prime }\cup
i=N$, or to any other indicator of how important the behavior or presence of
player~$i$ might be to the members of $N$ or an outside observer.

TU values and {extended value}s are defined on two distinct domains, $%
\mathcal{G}$ and $\mathcal{PG}$. {Extended value}s can be regarded as
technically the more general concept because any given TU value can be
turned into an {extended value} simply by ignoring the distribution $P$ that
is specified as part of probabilistic game $(N,v,P)$. For instance, the in
this way `generalized' Shapley value is defined by\footnote{%
When the considered set of players $N$ is clear from the context, we
simplify notation by writing $\sum_{S\not\ni i}$ instead of $%
\sum_{S\subseteq N:\, i\notin S}$, or $\sum_{S\ni i}$ instead of $%
\sum_{S\subseteq N:\, i\in S}$.}
\begin{equation}
\varphi_i(N,v,P)=\sum_{S\not\ni i}\frac{|S|!(|N|-|S|-1)!}{|N|!}(v(S\cup
i)-v(S)).
\end{equation}
and similarly the (generalized) Banzhaf value can be defined by
\begin{equation}
\beta_i(N,v,P)=\frac{1}{2^{n-1}}\sum_{S\not\ni i}(v(S\cup i)-v(S)).
\end{equation}

Both the original Shapley TU value and the Banzhaf TU value (which was at
first restricted to simple games, and later extended to general TU~games by \citeNP{Owen:1975:Banzhaf}) are
special instances of \emph{probabilistic values}, as introduced by \citeN{Weber:1988}, with either
\begin{equation}  \label{def:prob_value}
\Psi_i(N,v,Q)=\sum_{S\ni i}Q_i(S)(v(S)-v(S\setminus
i))=\mathbb{E}_{Q_i}[v(S)-v(S\setminus i)]\end{equation}
such that each element $Q_i$ of the collection $Q=\{Q_i\}_{i\in N}$ denotes a probability distribution on $%
\{S\subseteq 2^N\colon i\in S\}$, or
\begin{equation} \label{def:prob_value2}
\Psi_i(N,v,Q^{\prime })=\sum_{S\not\ni i}Q^{\prime }_i(S)(v(S\cup i)-v(S))=\mathbb{E}_{Q^{\prime }_i}[v(S\cup
i)-v(S)]
\end{equation}
such that $Q^{\prime }_i$ denotes a probability distribution on $2^{N\setminus i}$. For instance, %
\citeN{Laruelle/Valenciano:2005} have proposed two probabilistic values, $%
\Phi^+$ and $\Phi^{-}$, which respectively take $Q_i(S)$ and $Q_i^{\prime
}(S)$ to denote the probability of coalition $S$ being realized conditional
on $i$ voting \emph{no} and conditional on $i$ voting \emph{yes}.

For a given probabilistic game $(N,v,P)$ this suggests to work with the conditional probability distributions $%
P|i$ and $P|\neg i$ as follows: for all $S\subseteq N$
\begin{equation}
P| i(S) =
\begin{cases}
\frac{P(S)}{\sum\limits_{T\ni i}P(T)} & \text{if $i\in S$ and $%
\sum\limits_{T\ni i}P(T)\neq0$,} \\
0 & \text{otherwise,}%
\end{cases}%
\end{equation}
and, similarly,
\begin{equation}
P|\neg i(S) =
\begin{cases}
\frac{P(S)}{\sum\limits_{T\not\ni i}P(T)} & \text{if $i\notin S$ and $%
\sum\limits_{T\not\ni i}P(T)\neq0$,} \\
0 & \text{otherwise.}%
\end{cases}%
\end{equation}
One might then consider
\begin{equation}  \label{eq:Phi+}
\Phi^+_i(N,v,P)=\mathbb{E}_{P|i}[v(S)- v(S\setminus i)]
\end{equation}
and
\begin{equation}  \label{eq:Phi-}
\Phi^-_i(N,v,P)=\mathbb{E}_{P|\neg i}[v(S\cup i)- v(S)]
\end{equation}
as the natural extensions to Laruelle and Valenciano's conditional
decisiveness measures to domain $\mathcal{PG}$. Note that $%
\Phi^+(N,v,P)=\Phi^-(N,v,P)=\beta(N,v,P)$ if and only if $P(S)\equiv
2^{-|N|} $. One can similarly obtain identity with the (extended) Shapley
value: namely
\begin{eqnarray}  
&&\Phi^+(N,v,P)=\varphi(N,v,P) \nonumber\\\ \Longleftrightarrow \ &&P(\varnothing)=0 \text{
and } P(S)=\frac{1}{s \binom{n}{s} \sum_{t=1}^{n}\frac{1}{t}} \text{if $%
S\neq \varnothing$,}\label{eq:SV_as_Phi+}
\end{eqnarray}
and
\begin{eqnarray}
&&\Phi^-(N,v,P)=\varphi(N,v,P) \nonumber\\\ \Longleftrightarrow \ &&P(N)=0 \text{ and }
P(S)=\frac{1}{(n-s) \binom{n}{s} \sum_{t=1}^{n}\frac{1}{t}} \text{if $S\neq
N $,}\label{eq:SV_as_Phi-}
\end{eqnarray}
with $n=|N|$ and $s=|S|$ (see \citeNP[Prop. 3]{Laruelle/Valenciano:2005}).%
\footnote{%
Note that, as emphasized by Laruelle and Valenciano, the respective
distribution $P$ which needs to be assumed in order to obtain the Shapley
value as the expected marginal contribution conditional on $i$ being a
member of the random coalition $S$ and, alternatively, conditional on $i$
not being a member, \emph{differ}.}

We, however, suggest an altogether different approach to assessing the
importance of $N$'s members in a probabilistic games $(N,v,P)$. It is \emph{%
not} based on probabilistic values, nor marginal contributions in general.

The reason why weighted marginal contributions may misrepresent $(N,v,P)$ is that they
implicitly treat $i$'s decision, say, to change her \emph{no} vote into a
\emph{yes} (or vice versa) as being fully detached from the respective
probabilities of observing the considered two coalitions with and without $i$. Example
\ref{example1} already highlighted the effect that non-zero marginal contributions
$v(S\cup i)-v(S)>0$ simply don't count at all when the underlying probability distribution
$P$ treats both events $S\cup i$ and $S$ as null events. Adding up weighted marginal contributions also leads to strange conclusions if only one of the coalitions $S$ and $S\cup i$ has positive probability, as
illustrated in the following example.

\begin{example}\label{example2}
Consider 
an assembly of $3$ voters in which coalitions $\{1,3\}$, $\{2,3\}$ and $\{1,2,3\}$ are winning. Assume voters $2$ and $3$ are enemies and always vote contrary
to each other. Here, coalition $S=\{1,2\}$ might have positive probability under $P$ and
$P|\neg 3$, while $P(N)=0$. The problem with measures like $\Phi^-(N,v,P)$ is then that they are
strictly increased by a contribution which $3$ makes in the null event of joining $S=\{1,2\}$.
\end{example}

One thing that outside observers, members $j\neq i$ of $N$, or $i$ herself
might still care about is the informational gain that comes with the
knowledge: ``$i$ will (not) be part of the eventually formed coalition''.
Knowing this might imply that $j$ cannot (or must) be amongst the members of
the coalition. And it may have ramifications for the expected surplus that
is created or the passage probability of the bill being debated. In other
words, it may be useful to base one's evaluation of collective decision
making as described by $(N,v,P)$ on $P|i$ rather than $P$ when $i$ is known
to support the decision. This suggests looking at the difference $\mathbb{E}%
_{P|i}[v(S)]- \mathbb{E}_{P}[v(S)]$ as a way of quantifying $i$'s effect on
the outcome. And, of course, it is of similar interest~-- and may yield a
rather different quantification of the difference that $i$'s decision
makes~-- not to look at how much $i$'s support increases the expected worth $%
v(S)$ but at how much $i$'s opposition lowers it, i.e., $\mathbb{E}%
_{P}[v(S)]- \mathbb{E}_{P|\neg i}[v(S)]$. Combining these two evaluations of
how knowledge of $i$'s decision changes the expectation of the game by
summing them, we obtain:

\begin{definition}
\label{def:cpi} The \emph{prediction value~(PV)} of player~$i$ in the
probabilistic game $(N,v,P)$ is defined as
\begin{eqnarray}
\xi_i(N,v,P)&=&\mathbb{E}_{P|i}[v(S)]- \mathbb{E}_{P|\neg i}[v(S)]
\label{eq_cpi} \\
&=& \sum\limits_{S\ni i} v(S)\cdot P\!\!\mid\!\! i(S)- \sum\limits_{T\not\ni
i} v(T)\cdot P|\neg i(T).  \notag
\end{eqnarray}
\end{definition}


\begin{example}[Example \protect\ref{example1} revisited]
Consider again the canonical simple majority decision rule with an assembly
of 5 voters with $P(S)=0$ for $|S|=2$ or $|S|=3$ and $P(S)=1/12$ otherwise.
The conditional probabilities are given by
\begin{equation}
P|i(S)=%
\begin{cases}
\frac{1}{6} & \text{if $S=\{i\}\text{ or } S=N\backslash j,j\neq i \text{ or }S=N$,} \\
0 & \text{otherwise,}%
\end{cases}%
\end{equation}%
and, similarly,
\begin{equation}
P|\lnot i(S)=%
\begin{cases}
\frac{1}{6} & \text{if $S=$}\varnothing \text{ or }S=\{j\},j\neq i\text{ or }S=N%
\backslash i, \\
0 & \text{otherwise.}%
\end{cases}%
\end{equation}%
The prediction value follows as
\begin{eqnarray}
\xi _{i}(N,v,P) &=&\mathbb{E}_{P|i}[v(S)]-\mathbb{E}_{P|\lnot i}[v(S)] \\
&=&\sum\limits_{S\ni i}v(S)\cdot P\!\!\mid \!\!i(S)-\sum\limits_{T\not\ni
i}v(T)\cdot P|\lnot i(T)  \notag \\
&=&\frac{5}{6}-\frac{1}{6}=\frac{2}{3}.
\end{eqnarray}
\end{example}

\begin{remark}\label{remark1}
In case that coalition membership is statistically independent for every $%
i\neq j$, i.e., if $P$ is a product measure on $2^N$, the equality $%
P\!\!\mid\!\! i(S)=P\!\!\mid\!\! \neg i(S\setminus i)$ holds whenever $i\in
S $. Then equations (\ref{eq:Phi+}), (\ref{eq:Phi-}), and (\ref{eq_cpi}) all
evaluate to the same number~-- to the Banzhaf value, for instance, if $%
P(S)\equiv 2^{-|N|}$. That the ``expectation of a difference'' in (\ref{eq:Phi+})
or (\ref{eq:Phi-}) coincides with the ``difference between two
expectations'' in (\ref{eq_cpi}), however, fails to hold in general. In particular, we will
show in Corollary~\ref{cor2} that there is no probability distribution $P$ which would allow the Shapley value to be
interpreted as measuring informational importance.
\end{remark}

\section{Characterizing the {prediction value}\label{sec:characterization}}

This section provides an axiomatic characterization of the prediction value.
We begin with two classical conditions that are part of many axiomatic systems
in the literature on TU values. The first is \emph{anonymity},
which requires that the indicated difference to the game that is ascribed to
any player by an {extended value} does not depend on the labeling of the
players. The second is \emph{linearity}, which demands of an {extended value}
that it is linear in the characteristic function component $v$ of
probabilistic games.

\begin{definition}
Consider two probabilistic games $G=(N,v,P)$ and $G^{\prime }=(N^{\prime
},v^{\prime },P^{\prime })$ related through a bijection $\pi\colon
N\rightarrow N^{\prime }$ such that for all $S\subseteq N$, $v(S)=v^{\prime
}(\pi S)$ and $P(S)=P^{\prime }(\pi S)$ where $\pi S:=\{\pi(i)|i\in S\}$. An
{extended value} $\varphi$ is \emph{anonymous} if for every such $G$ and $%
G^{\prime }\in \mathcal{PG}$
\begin{equation}
\varphi_i(N,v,P)=\varphi_{\pi(i)}(N^{\prime },v^{\prime },P^{\prime }) \text{
for all }i\in N.
\end{equation}
\end{definition}

\begin{definition}
An {extended value} $\varphi$ is \emph{linear} if for all $%
(N,v,P),(N,v^{\prime },P)\in\mathcal{PG}$ and real constants $\alpha,\beta$
\begin{equation}
\varphi \left(N, \alpha v+\beta v^{\prime },P\right) =\alpha \varphi
\left(N, v,P\right) +\beta \varphi \left(N, v^{\prime },P\right) .
\end{equation}
\end{definition}

Linearity combines two properties, \emph{scale invariance} and \emph{additivity}. Especially the latter is far from being innocuous.\footnote{See, e.g., \citeN[6.2.26]{Felsenthal/Machover:1998} and \citeN[p.~248]{Luce/Raiffa:1957}.} But linearity is frequently imposed on solution concepts for TU games; and the PV, as the difference of two expectations, embraces it rather naturally.

The third characteristic property of the {PV} concerns the way how the
respective {extended value}s of two games $G$ and $G^{\prime }$ compare when
one can be viewed as a reduced form of the other. We first formalize this
reduction relation between two games, and afterwards define a consistency
property which connects the {extended value}s of correspondingly related
games.

\begin{definition}
\label{def:dependent_and_reduced_game} Call player $i\in N$ \emph{dependent}
in $(N,v,P)$ (or simply in $v$) if $v(i) =0.$ Given $G=(N,v,P)\in\mathcal{PG}
$ and a dependent player $i\in N$, the probabilistic game $%
G_{-i}=(N_{-i},v_{-i},P_{-i})\in\mathcal{PG}$ is a \emph{reduced game}
derived from $G$ by removal of $i$ if
\begin{align}  \label{eq vmini}
N_{-i}&=N\setminus i, \\
\noalign{\vskip6pt} P_{-i}(S)&=P(S)+P(S\cup i) \text{ for all $S\subseteq
N\setminus i$} \text{, and} \\
v_{-i}\left( S\right) &=
\begin{cases}
\frac{P(S)}{P(S)+P(S\cup i)}\cdot v( S) +\frac{P(S\cup i)}{P(S)+P(S\cup i)}
v( S\cup i) & \text{if }P_{-i}(S) >0,\footnotemark  \\
0 & \text{otherwise}.%
\end{cases}%
\end{align}
\end{definition}

\footnotetext{%
Note that if $i$ were not a dependent player, i.e., $v(i)\neq 0$, then $%
v_{-i}$ would not be a properly defined TU game because $v_{-i}(\varnothing)%
\neq 0$ in this case.} So, when one moves from a given probabilistic game $G$
to the reduced game $G_{-i}$, first, player~$i$ is removed from the set of
players; second, the probabilities of all coalitions in $G$ which only
differ concerning $i$'s presence are aggregated; and, third, the
corresponding new worth $v_{-i}(S)$ of coalitions $S\subseteq N_{-i}$ is the
convex combination of the associated old worths, $v(S)$ and $v(S\cup i)$,
weighted according to their respective probabilities under $P$. The
following property requires that the {extended value} of any player~$j\in
N_{-i} $ stays unaffected by the removal of $i$.\footnote{%
The condition is vaguely reminiscent of the \emph{amalgamation} properties
considered by \citeN{Lehrer:1988} or \citeN{Casajus:2012}.}

\begin{definition}
An {extended value} $\varphi$ is \emph{consistent} if for all $G=(N,v,P)\in%
\mathcal{PG}$ and all dependent players $i\in N$ in $v$, we have $%
\varphi_j(G)=\varphi_j(G_{-i})$
for all $j\in N\setminus i$.
\end{definition}

One reason for why this consistency property could be desirable is the following. Suppose that the considered model is misspecified in the sense that
a player of interest in the game is not taken into account by the rest of
the players (or an outside observer). For instance, consider the situation
of a voting game $G^{\prime }=(N^{\prime }, v^{\prime }, P^{\prime })$,
where the presence of a lobbyist~$i$ has been neglected. The more accurate
model would include the lobbyist and be $G=(N^{\prime }\cup i, v, P)$. The
effect of the lobbyist endorsing a proposal or opposing it would explicitly
be captured by the probability distribution $P$: for example, voters with
strong ties to $i$ may be likely to vote the same way, while others behave
oppositely. Coalitions $S$ and $S\cup i$ which differ only in $i$'s presence
will consequently have very different $P$-probabilities depending on whether
$S$ includes $i$'s fellow travelers or opponents. But if the probability $%
P^{\prime }$ and value $v^{\prime }$ of each coalition $T\subseteq N^{\prime
}$ in the `misspecified' game without $i$ are defined in a
probabilistically correct way, i.e., if the misspecified game $G^{\prime }$
equals $G_{-i}$, then the assessment of any actor $j\neq i$ should be
unaffected by whether one considers $G$ or $G_{-i}$.

Consistency can thus be seen as formalizing robustness to probabilistically
correct misspecifications.

\begin{proposition}\label{propALC}
The {prediction value} is anonymous, linear, and consistent.
\end{proposition}

\begin{proof}
Anonymity and linearity of $\xi$ are obvious from Definition~\ref{def:cpi}.
To prove consistency, consider $(N,v,P)\in\mathcal{PG}$ and let $i\in N$ be
dependent in $v$. Let $j\in N\setminus i$ and $S\subseteq N\setminus ij$. In
case $\sum\limits_{T\subseteq N\setminus i:\, j\in
T}P_{-i}(T)=\sum\limits_{T\subseteq N:\, j\in T}P(T)\neq 0$ we can compute
\begin{eqnarray*}
P_{-i}|j\left( S\cup j\right) &=&\frac{P_{-i}\left( S\cup j\right)}{%
\sum_{T\subseteq N\setminus i:\,T\ni j}P_{-i}\left( T\right) } =\frac{%
P\left( S\cup j\right) +P\left( S\cup i\cup j\right) }{\sum_{T\subseteq
N\setminus i:\,T\ni j}\left\{ P\left( T\right) +P\left( T\cup i\right)
\right\} }  \notag \\
&&=\frac{P\left( S\cup j\right) +P\left( S\cup i\cup j\right) }{\sum_{T
\subseteq N:\, T \ni j}P\left( T\right) }=P|j\left( S\cup j\right)
+P|j\left( S\cup i\cup j\right).
\end{eqnarray*}
In the alternative case, both sides are zero by definition. So in either
case
\begin{equation}
P_{-i}|j\left( S\cup j\right)= P|j\left( S\cup j\right) +P|j\left( S\cup
i\cup j\right).  \label{eq pmini}
\end{equation}
Analogously, one can check that
\begin{equation}
P_{-i}|\lnot j\left( S\right)= P|\lnot j\left( S\right) +P|\lnot j\left(
S\cup i\right)  \label{eq pmini2}
\end{equation}
when $S\subseteq N\setminus ij$. By using the definition of $v_{-i}$ and
invoking equality $(\ref{eq pmini})$ one can verify that
\begin{equation}
P_{-i}|j\left( S\cup j\right) v_{-i}\left( S\cup j\right) = P|j\left(
S\cup i\cup j\right) v\left( S\cup i\cup j\right) + P|j\left( S\cup j\right)
v\left( S\cup j\right).  \label{eq cons0}
\end{equation}%
Similarly, by definition of $v_{-i}$ together with $(\ref{eq pmini2})$, we
get%
\begin{equation}  \label{eq cons1}
P_{-i}|\lnot j\left( S\right) v_{-i}\left( S\right) = P|\lnot j\left( S\cup
i\right) v\left( S\cup i\right) + P|\lnot j\left( S\right) v\left( S\right).
\end{equation}
One can then infer
\begin{align*}
\xi_j(N_{-i},v_{-i},P_{-i}) &=\sum_{S\subseteq N\setminus ij}\big\{ %
P_{-i}|j\left( S\cup j\right) v_{-i}\left( S\cup j\right) - P_{-i}|\lnot
j\left( S\right) v_{-i}\left( S\right) \big\} \\
&=\sum_{S\subseteq N\setminus ij}\Big[ \big\{ P|j\left( S\cup i\cup j\right)
v\left( S\cup i\cup j\right) +P|j\left( S\cup j\right) v\left( S\cup
j\right) \big\}\\
& \qquad\qquad\quad -\big\{ P|\lnot j\left( S\cup i\right) v\left( S\cup
i\right) +P|\lnot j\left( S\right) v\left( S\right) \big\} \Big] \\
&=\sum_{S\subseteq N\setminus j}\big\{ P|j\left( S\cup j\right) v\left(
S\cup j\right) -P|\lnot j\left( S\right) v\left( S\right) \big\} \\
&= \xi_j(N,v,P),
\end{align*}%
where the second equality uses $(\ref{eq cons0})$ and $(\ref{eq cons1})$,
and the third one follows by shifting the corresponding terms from inside
the square brackets to the outer summation.
\end{proof}

Proposition \ref{propALC} is not enough to fully characterize the PV. For example,
$\Phi_i^+(N,v,P)$ satisfies anonymity, linearity and consistency, too (see Lemma~\ref{nonredundant} below). Theorem~\ref{thm:main} will provide a unique
characterization of the PV. In a nutshell, the underlying argument will be as follows: if extended values are linear and consistent, they are determined by their image for the subclass of $2$-player probabilistic games. It is then a question of how the extended values of
2-player probabilistic games should suitably be restricted.

For that purpose it is worth recalling two implications
of $i$ being part of the formed coalition: first, $i$'s presence means that $i$
contributes to the formed
coalition her voting weight, productivity, etc. This reveals information
about the expected worth directly. But, second, $i$'s presence also affects
the expected worth indirectly because it reveals information about the
presence and contributions of other players, at least if the behavior of $%
N\setminus i$ and of $i$ are not statistically independent. In case of
independence, i.e., if the presence of $i\in N$ presence has \emph{no}
informational value according to $P$, and if moreover $i$ is a null player
in the TU-game $(N,v)$, then a reasonable {extended value} can be expected
to assign zero to $i$. If, in contrast, knowledge of the behavior of null
player $i$ does change the odds of a proposal being passed, then $i$ has
positive informational value.

For illustration, consider a voting game in which $j$ is a dictator
according to the rules formalized by $v$ (i.e., $v(S)=1\Leftrightarrow j\in
S $). Let the voting behavior of $j$ be perfectly correlated with that of
some other player~$i$ (formally a null player). Now note that it is not part
of the model $(N,v,P)$, which mathematically describes the rules of the
collective decision body involving $i$ and $j$ and the random \emph{outcomes}
of coalition formation processes, \emph{why} the votes of $i$ and $j$ always
coincide. `Null player'~$i$ might simply follow `dictator'~$j$ in all
his decisions. Alternatively, player~$i$ could be irrelevant merely from a
formal perspective, i.e., have no say \emph{de jure}; while it is her who
imposes all her wishes on $j$~-- that is, she rules \emph{de facto}. In
either case the informational values of $i$ and $j$ are identical. They are
also maximal (and could plausibly be normalized to, say, 1) in the sense
that the outcome can be predicted perfectly when knowing that $i$ or $j$
votes \emph{yes} or \emph{no}.

We combine the requirement that an independent null player $i$ should be
assigned an {extended value} of zero with the requirement that $i$ has a
value of one in the considered perfect correlation case as follows:\footnote{%
The case of independence corresponds to $P|i(ij)=P|\neg i(j)$, while the
correlated dictator case amounts to $P|i(ij)=1$ and $P|\neg i(j)=0$.}

\begin{definition}\label{defIDDP}
An {extended value} $\varphi $ satisfies the \emph{informational
dummy-dictator property (IDDP)} if for $i\in N$ and $|N|=2$
\begin{equation}
\varphi _{i}(\{i,j\},u_{j},P)=P|i(ij)-P|\lnot i(j).
\end{equation}
\end{definition}


Regarding dictators themselves it makes sense to impose the following for
$1$-player probabilistic games:

\begin{definition}\label{defFULLCONTROL}
An {extended value} $\varphi$ satisfies \emph{full control} if $%
\varphi_i(\{i\},u_i,P)=1$ for all $i,P$ where $P(\{i\})>0$,
and $\varphi_i(\{i\},u_i,P)=0$ otherwise.
\end{definition}
This formalizes that if $N$ consists of just a single player~$i\in
\mathbb{N} $ with $v(i)=u_i(i)=1$ then $i$'s importance or the difference
that $i$ makes to this game should plausibly be evaluated as unity.\footnote{%
One might actually debate whether this should also be required in case that $%
P(\varnothing)=1$.} Immediately from the definition of the PV we obtain

\begin{proposition}\label{propFC_IDDP}
The {prediction value} satisfies full control and (IDDP).
\end{proposition}

\begin{remark}\label{remark2}
We remark that (IDDP) implies a positive {extended value} for a null player~$%
i$ even if $i$'s behavior is imperfectly but still positively correlated
with that of a dictator $j$. This is, e.g., the case when a \emph{yes}-vote
by $i$ is made more likely by most other players voting \emph{yes}, i.e.,
for the implicit probabilistic model behind the Shapley value. For a
probabilistic game with a dictator where $P$ reflects any Shapley value-like
probabilistic assumptions, this means that {PV} and Shapley value $\varphi$
will \emph{not} coincide: the Shapley value satisfies the traditional \emph{%
null player axiom}, i.e., it assigns zero to any player~$i$ who does not directly
affect the worth of any coalition $S$.
\end{remark}

We have the following characterization result:
\begin{theorem}
\label{thm:main} There is a unique extended value $\varphi $ which satisfies linearity, consistency, full control and (IDDP). It is anonymous and $\varphi\equiv
\xi$.
\end{theorem}
The full proof is provided in the appendix, together with proof of the following lemma. It certifies that none of the four axioms in Theorem~\ref{thm:main} is redundant.

\begin{lemma}\label{nonredundant}
$\text{ }$
\begin{enumerate}
\item[(i)] The extended value
\begin{equation*}
\Psi^1_i(N,v,P)=\Phi_i^+(N,v,P)
\end{equation*}
satisfies linearity, consistency, full control but not (IDDP).

\item[(ii)] The extended value
\begin{equation*}
\Psi^2_i(N,v,P)=\xi_i(N,v,P)-\Phi_i^+(N,v,P)
\end{equation*}
satisfies linearity, consistency, (IDDP) but not full control.

\item[(iii)] The extended value
\begin{equation*}
\Psi^3_i(N,v,P)=\sum_{S\ni i:|S|\le 2}v(S)\cdot P|i(S)-\sum_{T\not\ni
i:|T|\le 2}v(T)\cdot P|\neg i(T)
\end{equation*}
satisfies linearity, full control, (IDDP) but not consistency.

\item[(iv)] Let $|N|\ge 3$ and $v=\sum\limits_{S\subseteq N} \alpha_S\cdot u_S$ be the
unique decomposition of $v$ into unanimity games. The extended value
\begin{equation*}
\Psi^4_i(N,v,P)=\sum\limits_{S\subseteq N:\alpha_S\neq 0} \xi_i(N,u_S,P)
\end{equation*}
satisfies consistency, full control, (IDDP) but not linearity.
\end{enumerate}
\end{lemma}
%


\section{Relation between prediction value and probabilistic values \label%
{sec:relation_to_prob_values}}

\noindent The example values discussed earlier (like $\varphi,\beta,\Phi^+,%
\Phi^-$) all are members of the class of probabilistic values, i.e., they
have in common that they weight marginal contributions of a player by some
probability measure. We already noted in Remark \ref{remark1} that the natural extension of the
Banzhaf value agrees with the prediction value if $P(S)\equiv 2^{-|N|}$. We
now study the relationship between such members of the class of \emph{%
extended} probabilistic values and the prediction value somewhat more
generally.

Recall that \citeN{Weber:1988} has shown that the class of probabilistic
values is characterized by \emph{linearity, positivity,} and the \emph{null
player axiom}. The {PV} is linear but satisfies neither positivity nor the
null player axiom.\footnote{See Remark \ref{remark2} concerning null players.
Positivity is violated, e.g., for a probabilistic game $(N,v,P)$ where $v$
is positive and there is a player $i$ such that $P|i\equiv 0$.} Like the
prediction value, probabilistic values generally do not satisfy \emph{%
symmetry}. That property formalizes the idea that any symmetric players $%
i,j\in N$ in a TU game $(N,v)$ should have the same value; it will only be
satisfied by a probabilistic value $\Psi$ or the PV $\xi$ if the respective probability
measures $Q$ from \eqref{def:prob_value} or $P$ from \eqref{eq_cpi} are fully symmetric regarding $i$ and $j$.

The following result characterizes the connection between probabilistic values and the
prediction value.

\begin{theorem}
\label{thm_prob_val_eq_cpi} The identity $\Psi(\cdot,Q)\equiv \xi(\cdot,P)$ holds for $n>1
$ if and only if there exist probabilities $0< \tilde{p}_i<1$ for each player such that
\begin{equation}\label{theoremprobPV}
Q_i(S\cup i)=P|i(S\cup i)=\prod_{j\in S}\tilde{p}_j\cdot\prod_{j\in
N\backslash (S \cup i)} \left(1-\tilde{p}_j\right)
\end{equation}
holds for all $S\subseteq N\backslash i$, $i\in N$ and
\begin{equation}\label{eq:PS_in_thmprobvaleqcpi}
P(S)=\prod_{j\in S}\tilde{p}_j\cdot\prod_{j\in N\backslash S} \left(1-\tilde{%
p}_j\right)
\end{equation}
holds for all $S\subseteq N$.
\end{theorem}

The proof can be found in the appendix.

An important subclass of probabilistic values has the symmetry property:
\emph{semivalues} are defined by (\ref{def:prob_value}) and weights $Q_i(S)$
that depend on $S$ only via $|S|$ \shortcite{Dubey/Neyman/Weber:1981}.%
\footnote{%
See, e.g., \citeN{Calvo/Santos:2000} for a discussion of other subclasses
like \emph{weighted Shapley values}, \emph{weak semivalues} ($Q_i(S)$
depends only on $S$), or \emph{weighted weak semivalues} ($Q_i(S)$ is
decomposable as $w_i\cdot p_S$ where $w_i$ depends only on $i$ and $p_S$
depends only on $S$).} They are defined by
\begin{equation}\label{unscaledsemi}
f_i^q(N,v)=\sum\limits_{S\subseteq N\backslash i}q_{|S|}\cdot\Big(v(S\cup
i)-v(S)\Big)
\end{equation}
for a vector of $n$~non-negative numbers $q=(q_0,\dots,q_{n-1})\neq 0$ with
\begin{equation}
\label{eq_sv_normalization}
\sum\limits_{k=0}^{n-1}{\binom{{n-1}}{k}}q_k=1.
\end{equation}
The Shapley value arises by setting $q_k=\frac{1}{n{\binom{{n-1}}{k}}}$; the
Banzhaf index for $q_k=\frac{1}{2^{n-1}}$.

The following result characterizes the connection between semivalues and the
prediction value by answering the question: for which $q$ can one find $P$
such that $f^q(N,v)=\xi(N,v,P)$ for all $(N,v)\in\mathcal{G}^N$? This
identifies all semivalues which can be interpreted as the
prediction value for specific $P$.

\begin{proposition}
\label{prop_par} For a given semivalue $f^q$
and $n>1$ there exists $P$ such
that $f^q(\cdot)\equiv \xi(\cdot, P)$ on $\mathcal{G}^N$ if and only if there
is an $\alpha>0$ with $q_k=q_0\alpha^k>0$ for all $0\le k\le n-1$, where
$q_0^{-1}=\sum_{k=0}^{n-1}{\binom{{n-1}}{k}}\alpha^k$.
\end{proposition}
\begin{proof}
Assume $f^q(\cdot)\equiv \xi(\cdot,P)$. Then, anticipating Lemma~\ref{lemma_c_1} (see the appendix), we conclude $q_{|S|-1}=P|i(S)$ for all
   $\{i\}\subseteq S\subseteq N$ and all $i\in N$.
Applying Theorem~\ref{thm_prob_val_eq_cpi} it follows that there exist $\tilde{p}_j\in(0,1)$ for all $j\in N$
  such that $q_{|S\backslash i|}=\prod_{j\in S}\tilde{p}_j\cdot\prod_{j\in N\backslash (S\cup i)}
   \left(1-\tilde{p}_j\right)$.
   From  $P|i(S)=q_{|S|-1}=P|j(S)$, where $i,j\in S$, one obtains $\tilde{p}_i=\tilde{p}_j$ for all $i,j\in N$. Setting
   $\alpha=\frac{\tilde{p}_1}{1-\tilde{p}_1}$ we can write $q_k=\alpha^k\cdot \left(1-\tilde{p}_1\right)^{n-1}=
   \tilde{p}_1^k\left(1-\tilde{p}_1\right)^{n-k-1}$ for all $0\le k\le n-1$. We observe
   that $q$ satisfies equation~(\ref{eq_sv_normalization}) if we choose
   $q_0^{-1}=\sum_{k=0}^{n-1}{\binom{{n-1}}{k}}\alpha^k=\left(1-\tilde{p}_1\right)^{-n+1}$.

   Since
$\tilde{p}_1\mapsto\frac{\tilde{p}_1}{1-\tilde{p}_1}$
   is a bijection from $(0,1)$ to $(0,\infty)$, we directly obtain $\tilde p_1$ from a given $\alpha>0$ and can then easily check that $f^q$ defined by $q_k=q_o\alpha^k$ is indeed identical to $\xi(\cdot,P)$ with $P$ defined by $\tilde p_j=\tilde p_1$ for all $j\in N$ and by equation (\ref{eq:PS_in_thmprobvaleqcpi}).
\end{proof}

It follows that semivalues which allow for the
interpretation as a prediction value form a special subclass of semivalues.
They are known as \emph{binomial semivalues} (see \shortciteNP{Dubey/Neyman/Weber:1981}; \citeNP{Carreras/Freixas:2008}; \citeNP{Carreras/Puente:2012}~-- an axiomatic characterization has been
given by \citeNP{Amer/Gimenez:2007}). Specifically, a $p$-binomial semivalue
is defined by
\begin{equation}
q_k=p^k\left(1-p\right)^{n-k-1}\text{~~for~~~}0<p<1.
\end{equation}
Setting $\alpha=\frac{p}{1-p}$ matches the parametrization in Proposition~\ref{prop_par}, where $q_0$ can be determined from equation~(\ref{eq_sv_normalization}).\footnote{%
In some definitions in the literature
the extreme cases $p=0$ and $p=1$ are allowed, too, with the convention $%
0^0=1$. For $p=0$ we would get
the \emph{dictatorial index} and for $p=1$ the \emph{marginal index}. See %
\citeN{Owen:1978} for details. However, note that neither
$p=0$ nor $p=1$ satisfy the conditions from Proposition~\ref{prop_par}.}
For each given $\alpha>0$ we obtain a unique
semivalue. The Banzhaf value corresponds to $\alpha=1$ and $q_0=\frac{1}{2^{n-1}}
$. For $n=1$ each $\alpha>0$ yields the same value given by $q=1$. For $n=2$ we
set $\alpha=\frac{q_1}{q_0}$. For $n\ge 3$ it depends on the specific semivalue
whether it can be viewed as a restriction of the PV or not. As already suggested by our Remark~\ref{remark2} on correlated decisions and null players, a negative result obtains for the \emph{Shapley value} $\varphi$. It illustrates the
fundamental difference between traditional semivalues and the new value concept
proposed in this paper:

\begin{corollary}\label{cor2}
For $n\ge 3$ there exists no $P$ such that $\varphi(\cdot)\equiv
\xi(\cdot,P) $ on $\mathcal{G}^N$.
\end{corollary}

\begin{proof}
Recall that $\varphi\equiv f^q$ with $q_k=\big[n{\binom{{n-1}}{k}}\big]^{-1}$.
Let $n\ge 3$ and $P$ be such that $\varphi(\cdot)\equiv \xi(\cdot,P)$.
We can
deduce $q_0=\frac{1}{n}$ and $\alpha=\frac{q_1}{q_0}=\frac{1}{n-1}$ from
Proposition~\ref{prop_par}. Since $q_2=q_1\cdot\frac{2}{n-2}$ the condition $%
q_2=q_0\alpha^2= q_1\alpha$ implies $n=0$, in contradiction to $n\ge 3$.
\end{proof}


\section{Prediction values in the Dutch Parliament 2008--2010}

\label{sec: DP}

\noindent As illustration of the {prediction value}'s practical
applicability and of how its informational importance indications can be
very different from power ascriptions by traditional values, we consider the
seat distribution and voting behavior in the Dutch Parliament between 2008
and 2010. This was the period of the left-centered \emph{Balkenende IV}
government, which consisted of Christian democrats from the CDA and Christen
Unie parties and the social democratic PvdA.

\begin{table}[htp]
\begin{center}
{\scriptsize \
\begin{tabular}{c|c|c|c|c|c|c|c|c|c|c|c|}
\multicolumn{1}{c}{} & \multicolumn{1}{c}{CDA} & \multicolumn{1}{c}{CU} &
\multicolumn{1}{c}{D66} & \multicolumn{1}{c}{GL} & \multicolumn{1}{c}{PvdA}
& \multicolumn{1}{c}{PvdD} & \multicolumn{1}{c}{PVV} & \multicolumn{1}{c}{SGP
} & \multicolumn{1}{c}{SP} & \multicolumn{1}{c}{Verdonk} &
\multicolumn{1}{c}{VVD} \\ \cline{2-12}
Seats & 41 & 6 & 3 & 7 & 33 & 2 & 9 & 2 & 25 & 1 & 21 \\ \cline{2-12}
$\beta$ & 0.597 & 0.073 & 0.038 & 0.089 & 0.398 & 0.026 & 0.120 & 0.026 &
0.306 & 0.013 & 0.200 \\ \cline{2-12}
$\varphi$ & 0.317 & 0.036 & 0.021 & 0.044 & 0.225 & 0.015 & 0.061 & 0.015 &
0.155 & 0.007 & 0.104 \\ \cline{2-12}
$\Phi^+$ & 0.665 & 0.040 & 0.005 & 0.051 & 0.283 & 0.004 & 0.074 & 0.004 &
0.235 & 0.001 & 0.210 \\ \cline{2-12}
$\Phi^-$ & 0.660 & 0.021 & 0.004 & 0.050 & 0.434 & 0.005 & 0.061 & 0.002 &
0.140 & 0.000 & 0.131 \\ \cline{2-12}
$\xi$ & 0.782 & 0.318 & 0.248 & 0.468 & 0.330 & 0.023 & 0.369 & 0.182 & 0.217
& 0.217 & 0.278 \\ \cline{2-12}
\end{tabular}
\medskip }
\end{center}
\caption{Values in the Dutch Parliament}
\label{tab:DP_values}
\end{table}
The distribution of the 150 seats in parliament between its eleven parties
is displayed in the top part of Table~\ref{tab:DP_values}. The three
government parties held a majority of 80 out of 150 seats. When voting on
non-constitutional propositions, the Dutch Parliament applies simple
majority rule. It is straightforward to define a voting game with this
information, and to calculate the corresponding \emph{a priori} Banzhaf and
Shapley values $\beta$ and $\varphi$.

We used the parliamentary information system \emph{Parlis}\footnote{%
The data is available through \texttt{http://data.appsvoordemocratie.nl}} in
order to extract information on members, meetings, votes and decisions on
propositions in the 2008--2010 period. From the records of regular plenary
voting rounds, where parties vote as blocks,
we derived the empirical frequencies of the $2^{11}$ conceivable divisions
into \emph{yes} and \emph{no}-camps from 2720 observations.\footnote{%
We pooled all regular plenary votes in order to illustrate the simplest way
in which data can be used to infer interdependencies in a voting body~-- one
might want to split the data with respect to topics, or weight distinct
calls by their importance, in actual political analysis. Note that the Dutch
Parliament's chairperson assumes that parties vote as blocks unless some MP
demands voting by call. Only then can members of the same party vote
differently. We excluded such cases of `non-coherent voting' from our
analysis.} Defining $P$ by these empirical frequencies, we calculated the
corresponding prediction values $\xi_i$ of the parties as well as their positive and negative conditional decisiveness values $\Phi^+_i$ and $\Phi^-_i$ defined in (\ref{eq:Phi+}) and (\ref{eq:Phi-}). A summary of the
results is given in the bottom part of Table~\ref{tab:DP_values}.

The {PV}-scores $\xi_i$ of Dutch parties tend to be higher than their
respective traditional Banzhaf or Shapley power measures $\beta_i$ and $%
\varphi_i$, and even the decisiveness measures $\Phi^+_i$ and $\Phi^-_i$
which incorporate the same empirical estimate of $P$. In particular, the {%
prediction value} ascribes rather substantial numbers also to small parties
like D66, SGP, or Verdonk.

\begin{table}[htp]
\begin{center}
{\scriptsize \
\begin{tabular}{r|r|r|r|r|r|r|r|r|r|r|r|}
\multicolumn{1}{c}{} & \multicolumn{1}{c}{CDA} & \multicolumn{1}{c}{CU} &
\multicolumn{1}{c}{D66} & \multicolumn{1}{c}{GL} & \multicolumn{1}{c}{PvdA}
& \multicolumn{1}{c}{PvdD} & \multicolumn{1}{c}{PVV} & \multicolumn{1}{c}{SGP
} & \multicolumn{1}{c}{SP} & \multicolumn{1}{c}{Verdonk} &
\multicolumn{1}{c}{VVD} \\ \cline{2-12}
CDA & 1.000 & 0.267 & 0.263 & 0.483 & 0.237 & -0.044 & 0.324 & 0.221 & -0.026
& -0.026 & 0.012 \\ \cline{2-12}
CU & 0.267 & 1.000 & 0.631 & 0.348 & 0.601 & 0.015 & 0.178 & 0.459 & 0.094 &
0.094 & 0.158 \\ \cline{2-12}
D66 & 0.263 & 0.631 & 1.000 & 0.348 & 0.811 & 0.044 & 0.169 & 0.693 & 0.034
& 0.034 & -0.008 \\ \cline{2-12}
GL & 0.483 & 0.348 & 0.348 & 1.000 & 0.315 & -0.003 & 0.171 & 0.259 & 0.019
& 0.019 & 0.068 \\ \cline{2-12}
PvdA & 0.237 & 0.601 & 0.811 & 0.315 & 1.000 & 0.040 & 0.161 & 0.714 & 0.027
& 0.027 & -0.003 \\ \cline{2-12}
PvdD & -0.044 & 0.015 & 0.044 & -0.003 & 0.040 & 1.000 & 0.198 & 0.171 &
0.536 & 0.536 & 0.389 \\ \cline{2-12}
PVV & 0.324 & 0.178 & 0.169 & 0.171 & 0.161 & 0.198 & 1.000 & 0.203 & 0.263
& 0.263 & 0.285 \\ \cline{2-12}
SGP & 0.221 & 0.459 & 0.693 & 0.259 & 0.714 & 0.171 & 0.203 & 1.000 & 0.110
& 0.110 & 0.025 \\ \cline{2-12}
SP & -0.026 & 0.094 & 0.034 & 0.019 & 0.027 & 0.536 & 0.263 & 0.110 & 1.000
& 1.000 & 0.554 \\ \cline{2-12}
Verdonk & -0.026 & 0.094 & 0.034 & 0.019 & 0.027 & 0.536 & 0.263 & 0.110 &
1.000 & 1.000 & 0.554 \\ \cline{2-12}
VVD & 0.012 & 0.158 & -0.008 & 0.068 & -0.003 & 0.389 & 0.285 & 0.025 & 0.554
& 0.554 & 1.000 \\ \cline{2-12}
\end{tabular}
\medskip }
\end{center}
\caption{Correlation coefficients for 2008--2010 votes in Dutch Parliament}
\label{tab:DP_correlation}
\end{table}

This reflects specificities of the political situation in the Netherlands
and that the {PV} picks up corresponding correlations between the voting
behavior of different parties. Varying majorities at calls are quite common
in the Dutch Parliament. The member parties of the government do not
necessarily vote the same way; some are frequently supported by smaller
opposition parties. The correlation coefficients reported in Table~\ref%
{tab:DP_correlation} indicate, for instance, that SGP and D66 quite commonly
voted the same way as CU and PvdA. Their PV numbers hence differ much less
than their seat shares.

Verdonk and SP constitute an extreme case in this respect. The former is
commonly considered as right-wing, the latter as a left-wing party;
still both voted the same way at each call in the data set (presumably
having different reasons). Perfect correlation of their votes implies that
both have the same prediction value~-- despite SP having 25 seats and
Verdonk but one: knowing either's vote in advance would have been equally
valuable for predictive purposes. Measures based on marginal contributions,
in contrast, clearly favor SP over Verdonk (though less so if the \emph{%
a~posteriori} correlation between SP's and Verdonk's votes is ignored).
Interestingly, the GL party has the second-highest prediction value: despite
it not being in government and having only the sixth-largest seat share,
support by GL was a better predictor of a bill's success than support by any
except the biggest party (CDA).


\section{Concluding Remarks\label{sec:conclusion}}

\noindent Traditional semivalues like the Shapley or Banzhaf values and the
prediction value provide two qualitatively distinct perspectives on the
importance of the members of a collective decision body. One highlights the
difference that an ad-hoc change of a given player~$i$'s membership in the
coalition which eventually forms would make from an ex ante perspective; the
other stresses the difference that the change of a player's presumed
membership makes for one's ex~ante assessment of realized worth. As the
figures in Table~\ref{tab:DP_values} illustrate, both can differ widely in
case players' behavior exhibits interdependencies. But, as captured by
Proposition~\ref{prop_par}, they coincide in case of statistical independence.
The latter is presumed by the behavioral model underlying, e.g., the Banzhaf
value, but incompatible with that underlying the Shapley value.

For independent individual voting decisions, the conditioning on different
votes of player~$i$ adds no behavioral information to the numerical one
about $i$'s weight contribution to either the \emph{yes} or \emph{no} camp. Then $i$'s
informational importance and $i$'s voting power or influence~-- reflected by
sensitivity of the collective decision to a last-minute change of $i$'s
behavior~-- are aligned.\footnote{In the case of the Banzhaf value, coincidence between voter~$i$'s influence
as picked up by $i$'s average marginal contribution and the informational
effect of knowing $i$'s vote has been hinted at by \citeN[3.2.12--15]{Felsenthal/Machover:1998}.}

It might be criticized that in cases of interdependence, the prediction
value fails to distinguish correlation and causation. For illustration,
consider decisions by a weighted voting body in which some player~$i$ has
zero weight but all other players' decisions are perfectly correlated with
that of $i$. Player~$i$'s prediction value is then one irrespective of
whether (i) players $j\neq i$ `follow' $i$ as, say, their guru or supreme
leader and cast their weight as $i$ would if he had any, (ii) $i\neq k$ and
all players $j\neq k$ follow a specific other player~$k$, or (iii) all
players debate the merit of a proposal based on different initial
inclinations and collective opinion dynamics converge to, for instance, the
majority inclination.\footnote{%
See \citeANP{Grabisch/Rusinowska:2010}'s \citeyear{Grabisch/Rusinowska:2010}
related work on possibilities to aggregate individual influence in command
structures.} But since knowing $i$'s \emph{decision}~-- rather than $i$'s initial
inclination~-- will always fully reveal the realized outcome, $\xi_i=1$ can
be regarded more as a feature than a flaw.

This example points to an interesting extension of the proposed ``difference
of conditional expected values''-approach to measuring importance. Namely,
start with a given description $(N,v,P)$ of a decision body where $P$
corresponds to, say, the Banzhaf uniform distribution and augment it by the
formal description of a social opinion formation process which defines a
mapping from players' binary initial voting inclinations to a distribution
over final ones after social interaction. One can then capture a player~$i$%
's combined social \emph{and} formal influence in the decision body by
answering the question: how much does knowing that $i$'s \emph{initial inclination}
is in favor (or against) modify the final outcome which is to be expected?
We conjecture that this approach actually has advantages over extending
marginal contribution-based analysis to social interaction,\footnote{%
See, for instance, the power scores derived from swings in societies with
opinion leaders by \shortciteN{vandenBrink/Rusinowska/Steffen:2013}.} and
plan to pursue this extension in future research.

\newpage


\section*{Appendix}

\subsection*{Proof of Theorem \ref{thm:main}}

\noindent The proof proceeds in three steps. First, in Lemma~\ref{lemma_char2p} we prove for $|N|=2$ that linearity and consistency imply that an extended value is determined by unanimity games. Second, we generalize this to all probabilistic games in Lemma~\ref{thm_fully_specified}. Finally, we show that the full control property and (IDDP) characterize the PV for 2-player probabilistic games and hence probabilistic games in general.

\begin{lemma}
\label{lemma_char2p} Consider an extended value $\varphi$ that is
linear on the space of all 2-player probabilistic games and consistent. For any set $N$ with $|N|=2$, the mapping $(N,v,P)\mapsto\varphi(N,v,P)$ is
fully determined by the numbers
\begin{equation}  \label{eq char2p}
x_{ij}:=\varphi_i(N,u_j,P)\ \mathnormal{\ for }\ i,j\in N.
\end{equation}
\end{lemma}

\proof Let $P$ be a fixed
probability distribution on $2^N$ with $N=\{i,j\}$. The set of unanimity games $%
\{u_{i},u_{j},u_{ij}\}$ forms a basis for the space of all TU games on $N$.
In particular, for any $(N,v)\in \mathcal{G}^N$ there are constants $%
\alpha_i,\alpha_j,\alpha_{ij}$ such that
\begin{equation}
v\equiv \alpha_iu_i+\alpha_ju_j+\alpha_{ij}u_{ij}.
\end{equation}
And thus, for arbitrary $P$ and $i\in N$, $\varphi$'s linearity implies
\begin{equation}
\varphi_i(N,v,P)=\alpha_i \underbrace{\varphi_i(N,u_i,P)}_{:=x_{ii}}+\alpha_j%
\underbrace{\varphi_i(N,u_j,P)}_{:=x_{ij}} +\alpha_{ij}\underbrace{%
\varphi_i(N,u_{ij},P)}_{:=x_{i,ij}}.
\end{equation}
We need to show that $x_{i,ij}$ and $x_{j,ij}$ are fully determined by $%
x_{ii}$ and $x_{ij}$.

To see this, notice first that both players are dependent in $(N,u_{ij},P)$. So
we may consider the reduced game obtained by $j$'s removal, which involves $%
N_{-j}=\{i\}$ and
\begin{align}
\begin{array}{ll}
P_{-j}(\varnothing)=P(\varnothing)+P(j), & \quad P_{-j}(i)=P(i)+P(ij), \\
(u_{ij})_{-j}(\varnothing)=0, & \quad (u_{ij})_{-j}(i)=
\begin{cases}
\frac{P(ij)}{P(i)+P(ij)} & \text{ if }P(i)+P(ij)>0, \\
0 & \text{ otherwise.}%
\end{cases}%
\end{array}%
\end{align}
In case $P(i)+P(ij)>0$, we have
\begin{align}
\varphi_{i}(N,u_{ij},P) &= \varphi_i\big(\{i\},\tfrac{P(ij)}{P(i)+P(ij)}%
\cdot u_{i} ,P_{-j}\big) \nonumber \\
&= \tfrac{P(ij)}{P(i)+P(ij)}\cdot \varphi_i(\{i\},u_{i},P_{-j}) \nonumber\\
&= \tfrac{P(ij)}{P(i)+P(ij)}\cdot \varphi_{i}(N,u_{i},P) = \tfrac{P(ij)}{%
P(i) +P(ij)}\cdot x_{ii},
\end{align}
where the first equality invokes consistency, the second linearity, and the
third one exploits that $(\{i\},u_{i},P_{-j})$ is the reduction of $%
(N,u_{i},P)$ by player~j and again consistency. When $P(i)=P(ij)=0$ we have $%
\varphi_{i}(N,u_{ij},P)=0$ because in this case $(u_{ij})_{-j}(\{i\})=0$ by
Definition~\ref{def:dependent_and_reduced_game}, so that $%
\left(u_{ij}\right)_{-j}$ is the all-zero game $\mathbf{0}$ in that case.
Consistency requires $\varphi_{i}(N,u_{ij},P)=\varphi_{i}(\{1%
\},(u_{ij})_{-j},P_{-j})= \varphi_{i}(\{i\},\mathbf{0},P_{-j})=0$ due to
linearity.

In summary,
\begin{equation}  \label{eq_varphi_u_12}
x_{i,ij}=
\begin{cases}
\tfrac{P(ij)}{P(i)+P(ij)}\cdot x_{ii} & \text{ if $P(i)+P(ij)>0$,} \\
0 & \text{otherwise.}%
\end{cases}%
\end{equation}
And in a similar fashion one obtains
\begin{equation}
x_{j,ij}=
\begin{cases}
\tfrac{P(ij)}{P(j)+P(ij)}\cdot x_{jj} & \text{if $P(j)+(Pij)>0$,} \\
0 & \text{otherwise}.%
\end{cases}%
\end{equation}
\endproof

For any $v\equiv \alpha_iu_i+\alpha_ju_j+\alpha_{ij}u_{ij}$ we have
\begin{equation}
\varphi_i(N,v,P)=
\begin{cases}
\alpha_j\cdot x_{ij}+\left(\alpha_i+\tfrac{\alpha_{ij}\cdot P(ij)}{P(i)+P(ij)%
}\right)\cdot x_{ii} & \text{ if $P(i)+P(ij)>0$}, \\
\alpha_j\cdot x_{ij}+\alpha_i\cdot x_{ii} & \text{otherwise}%
\end{cases}%
\end{equation}
and an analogous expression for $\varphi_j(N,v,P)$. This finding can be
generalized from just two players to arbitrary $N$:

\begin{lemma}
\label{thm_fully_specified} Let $\varphi$ be a consistent and
linear {extended value}. Then the mapping $(N,v,P)\mapsto \varphi(N,v,P)$ is
fully specified by the parameters in $(\ref{eq char2p})$.
\end{lemma}

\proof
Using the $n$-player unanimity games as a basis for $\mathcal{PG}^N$ one can
always write
\begin{equation}
v\equiv \sum\limits_{\varnothing\subsetneq T\subseteq N} \alpha_Tu_T.
\end{equation}
Letting $i\in N$ be an arbitrary but fixed player, we will use induction on
$n$ in order to prove the following

\medskip \noindent \textit{Claim:} There exist $\beta_{ij}$, depending on
the $\alpha_T$ and $P$, such that

\begin{equation}
\varphi_i(N,v,P)=\sum_{j=1}^{n} \beta_{ij}x_{ij} \text{ where } x_{ij}:=\varphi_i(N,u_j,P).
\end{equation}

\noindent The claim is obvious for a single player and was proven for $|N|=2$
in Lemma~\ref{lemma_char2p}. In view of linearity, it suffices to prove
the statement for unanimity games $u_{T}$, where nothing needs to be shown
when the cardinality of $T$ is one. So we consider $|N|\geq 3$, $|T|\geq 2$
and assume that the statement is true for all player sets $N$ of cardinality
$n-1$. Let $j\in N\setminus i$ be a player, which must be dependent in $%
u_{T} $ because $|T|\geq 2$. Now we consider the reduced game $%
(N_{-j},(u_{T})_{-j},P_{-j})$. From consistency we conclude
\begin{equation*}
\varphi _{i}(N,u_{T},P)=\varphi _{i}(N_{-j},(u_{T})_{-j},P_{-j}).
\end{equation*}%
Applying the induction hypothesis implies the existence of $\beta
_{ik}^{\prime }$, which depend on $P_{-j}$ and hence on $P$, such that
\begin{equation*}
\varphi _{i}(N,u_{T},P)=\sum_{k=1,k\neq j}^{n}\beta _{ik}^{\prime }\varphi
_{i}(N_{-j},u_{k},P_{-j}).
\end{equation*}%
Since $(u_{k})_{-j}=u_{k}$ the reduced game of $(N,u_{k},P)$ is given by $%
(N_{-j},u_{k},P_{-j})$ for all $1\leq k\leq n$ with $j\neq k$. Inserting $%
\varphi _{i}(N_{-j},u_{k},P_{-j})=\varphi _{i}(N,u_{k},P)=x_{ik}$ then
proves the claim, and the theorem. \endproof

We remark that the coefficients $\beta_{ij}$ referred to in the above proof
get quite complicated for increasing $n$. In the following we will use only
the fact that they are well-defined given $v$ and $P$.

\begin{proof}[Proof of Theorem \ref{thm:main}] 
To complete the proof we now show how the values $x_{ii}=\varphi_i(N,u_i,P)$ and $x_{ij}=%
\varphi_i(N,u_j,P)$ can be computed from the corresponding values for the
player set $N^{\prime }=\{i,j\}$. Since $\left(u_i\right)_{-j}=u_i$ for all $%
i\neq j$ we can recursively conclude from consistency
\begin{eqnarray}
\varphi_i(N,u_i,P)&=&\varphi_i(\{i,j\},u_i,P^\star) \text{ and} \label{eq_xii_red} \\
\varphi_i(N,u_j,P)&=&\varphi_i(\{i,j\},u_j,P^\star) ,  \label{eq_xij_red}
\end{eqnarray}
where
\begin{equation}
P^\star(S)=\sum_{T\subseteq N\setminus ij} P(S\cup T)\text{ for any }%
S\subseteq \{i,j\} .  \label{eq_P_star}
\end{equation}

Using equation~(\ref{eq_P_star}) and similarly defining
\begin{equation}
P^{\prime }(S)=\sum_{T\subseteq N\setminus i} P(S\cup T)\text{ for any }%
S\subseteq \{i\},
\end{equation}
we conclude $\varphi_i(\{i\},u_i,P^{\prime })=\varphi_i(\{i,j\},u_i,P^\star)$
from consistency. Thus, the full control property, in connection with
consistency and linearity, implies $x_{ii}=1$ for all player sets $N$
(containing player $i$). If $\varphi$ satisfies (IDDP) the values of $x_{ij}$ are determined, and hence $\varphi$ is determined on the class of
$2$-player probabilistic games. Then $\varphi\equiv
\xi$ follows from Lemma~\ref{thm_fully_specified}. Finally note that the full control property and (IDDP) do not depend on the labeling of the players, which implies anonymity.
\end{proof}
\subsection*{Proof of Lemma \ref{nonredundant}} $\text{ }$\\
\noindent
(i) Linearity of $\Phi ^{+}$ follows from (\ref{eq:Phi+}). For notational convenience put $\tilde{P}=P|i$. For the
reduced game $G_{-j}=(N_{-j},v_{-j},P_{-j})$ we get
\begin{eqnarray*}
&&\Phi _{i}^{+}(N_{-j},v_{-j},P_{-j})\\
 &=&\mathbb{E}_{\tilde{P}%
-j}[v_{-j}(S)-v_{-j}(S\setminus i)]=\sum_{S\subseteq N\backslash j}\tilde{P}%
_{-j}[v_{-j}(S)-v_{-j}(S\setminus i)] \\
&=&\sum_{S\subseteq N\backslash j}\left( \tilde{P}(S)+\tilde{P}(S\cup
j)\right) [v_{-j}(S)-v_{-j}(S\setminus i)] \\
&=&\sum_{S\subseteq N\backslash j}\left( \tilde{P}(S)+\tilde{P}(S\cup
j)\right) v_{-j}(S)-\sum_{S\subseteq N\backslash j}\left( \tilde{P}(S)+\tilde{P%
}(S\cup j)\right) v_{-j}(S\setminus i) \\
&=&\sum_{S\subseteq N\backslash j}\tilde{P}(S)v(S)+\tilde{P}(S\cup j)v(S\cup
j)-\sum_{S\subseteq N\backslash j}\left( \tilde{P}(S)v(S\setminus i)-\tilde{P}%
(S\cup j)v((S\cup j)\backslash i)\right)  \\
&=&\sum_{S\subseteq N\backslash j}\tilde{P}(S)\left[ v(S)-v(S\setminus i)%
\right] +\sum_{S\subseteq N\backslash j}\tilde{P}(S\cup j)\left[ v(S\cup
j)-v((S\cup j)\backslash i))\right]  \\
&=&\sum_{S\subseteq N}\tilde{P}(S)\left[ v(S)-v(S\setminus i)\right]
=\sum_{S\subseteq N}P|i\left[ v(S)-v(S\setminus i)\right]  \\
&=&\Phi _{i}^{+}(N,v,P).
\end{eqnarray*}
We conclude that $\Phi _{i}^{+}(N,v,P)$ is consistent.

\noindent The verification of \emph{full control} provides
\begin{eqnarray}
\Phi _{i}^{+}(\{i\},v,P) &=&\mathbb{E}_{P|i}[v(S)-v(S\setminus i)]
\nonumber  \\
&=&P|i\left( \{i\}\right) v(\{i\}) \label{FC}
\end{eqnarray}
which is equal to one for $v=u_{i}$ and $P\left( \{i\}\right) >0$ and
equal to zero if $P\left( \{i\}\right) =0$.

\noindent To see that $\Phi ^{+}$ does not satisfy (IDDP) note that
\begin{eqnarray}
\Phi _{i}^{+}(\{i,j\},v,P) &=&\mathbb{E}_{P|i}[v(S)-v(S\setminus i)] \nonumber  \\
&=&P|i\left( \{i,j\}\right) [v(\{i,j\})-v(\{j\})]+P|i\left( \{i\}\right)
[v(\{i\})-v(\varnothing )].
\end{eqnarray}
For the unanimity game $u_{j}$ follows{\ }%
\begin{equation}
\Phi _{i}^{+}(\{i,j\},u_{j},P)=0.  \label{notIDDP}
\end{equation}

\noindent (ii) $\Psi _{i}^{2}(N,v,P)$ inherits linearity and consistency from $\xi $ and $\Phi^{+}$. From (\ref{notIDDP}) follows
\begin{equation*}
\Psi _{i}^{2}(\{i,j\},u_{j},P)=\xi _{i}(\{i,j\},u_{j},P)
\end{equation*}
and therefore (IDDP). From Proposition~\ref{propFC_IDDP} and (\ref{FC}) 
 we know that both $\xi $ and $\Phi ^{+}$ satisfy full control such that
\begin{equation*}
\Psi _{i}^{2}(\{i\},u_{i},P)=0,
\end{equation*}
contrary to Definition \ref{defFULLCONTROL}.\\

\noindent (iii)  Linearity is obvious. For $|N|\leq 2$ the extended value $\Psi ^{3}$ is
identical to the PV and the latter satisfies full control and (IDDP). For a
counterexample to consistency consider a game $G_{-j}=(N,v,P)$ with $%
\left\vert N\right\vert =3$ and perfect correlation $P(N)=1/2=P(\varnothing
) $. Here,
\begin{equation}
\Psi _{i}^{3}(N,v,P)=0 \mbox{ for all } i\in N. \label{psi}
\end{equation}
However, for the
reduced game $G_{-j}=(N_{-j},v_{-j},P_{-j})$ we get
\begin{align}
N_{-j}& =N\setminus j, \nonumber \\
\noalign{\vskip6pt}P_{-j}(S)& =P(S)+P(S\cup j)\text{ for all $S\subseteq
N\setminus j$} \nonumber\\
& =1/2\text{ for }S\in \left\{ N\backslash j,\varnothing \right\} \text{ and
}0\text{ otherwise, } \nonumber\\
v_{-j}\left( S\right) & =
\begin{cases}
v(\varnothing ) & \text{for }S=\varnothing \nonumber\\
v(N) & \text{for }S=N\backslash j \nonumber\\
0 & \text{otherwise.}%
\end{cases}%
\end{align}
For $\Psi ^{3}$ follows
\begin{equation*}
\Psi _{i}^{3}(N_{-j},v_{-j},P_{-j})=v(N)-v(\varnothing )=v(N)  \mbox{ for all } i\in N
\end{equation*}
which does not coincide with \eqref{psi}.\\

\noindent (iv) Consider the reduced game $G_{-j}=(N_{-j},v_{-j},P_{-j})$. PV is consistent and therefore
\begin{equation*}
\xi _{i}(N_{-j},(u_{S})_{-j},P_{-j})=\xi _{i}(N,u_{S},P)\text{ for all }i\in
N\setminus j\text{.}
\end{equation*}
We conclude
\begin{eqnarray*}
\Psi _{i}^{4}(N_{-j},v_{-j},P_{-j}) &=&\sum\limits_{S\subseteq N:\alpha
_{S}\neq 0}\xi _{i}(N_{-j},(u_{S})_{-j},P_{-j}) \\
&=&\sum\limits_{S\subseteq N:\alpha _{S}\neq 0}\xi _{i}(N,u_{S},P)=\Psi
_{i}^{4}(N,v,P)\text{ for all }i\in N\setminus j
\end{eqnarray*}
which confirms consistency.

\noindent Full control and (IDDP) follows from $\Psi
_{i}^{4}(\{i\},u_{i},P)=\xi _{i}(\{i\},u_{i},P)$ and $\Psi
_{i}^{4}(\{i,j\},u_{j},P)=\xi _{i}(\{i,j\},u_{j},P)$.

\noindent To verify that $\Psi _{i}^{4}$ is not linear put $w=\sum\limits_{S\subseteq N}\beta _{S}\cdot
u_{S}$.
\begin{equation*}
\Psi _{i}^{4}(N,v+w,P)=\sum\limits_{S\subseteq N:\alpha _{S}+\beta _{S}\neq
0}\xi _{i}(N,u_{S},P)
\end{equation*}%
which is in general not equal to
\begin{equation*}
\sum\limits_{S\subseteq N:\alpha _{S}\neq 0}\xi
_{i}(N,u_{S},P)+\sum\limits_{S\subseteq N:\beta _{S}\neq 0}\xi
_{i}(N,u_{S},P).
\end{equation*}
\raggedleft{$\square$}

\raggedright

\subsection*{Proof of Theorem \ref{thm_prob_val_eq_cpi}}

The proof is based on three insights, stated in Lemmas \ref{lemma_c_1}--\ref{lemma_exclude_one}.

\begin{lemma}
\label{lemma_c_1} From $\Psi(\cdot,Q)\equiv \xi(\cdot,P)$ follows $Q_i(S)=P|i(S)$ for all $\{i\}\subseteq S\subseteq N$.
\end{lemma}

\begin{proof}
For an arbitrary subset $\{i\}\subseteq S\subseteq N$ we consider the
unanimity game $u_S$ and obtain the formulas
\begin{equation*}
\xi_i(u_S,P) = \sum\limits_{T\ni i} u_S(T)\cdot P\!\!\mid\!\! i(T)-
\sum\limits_{T \not\ni i} u_S(T)\cdot P\!\!\mid\!\! \neg i(T)
=\sum\limits_{T:S\subseteq T} P\!\!\mid\!\! i(T)
\end{equation*}
and
\begin{equation*}
\Psi_{i}(u_S,Q)=\sum\limits_{\{i\}\subseteq T\subseteq N}Q_i(T)\Big[%
u_S(T)-u_S(T\backslash i)\Big]=\sum\limits_{T:S\subseteq T} Q_i(T).
\end{equation*}
Now we prove the proposed statement by induction on the subsets $S$ in
decreasing order of their cardinalities using the assumption $%
\xi_i(u_S,P)=\Psi_{i}(u_S,Q)$. For the induction start $S=N$ we have $%
P|i(N)=Q_i(N)$. Using the induction hypothesis for all $S^{\prime
}\subseteq N$ with $|S^{\prime }|>|S|$ yields $P|i(S)=Q_i(S)$.
\end{proof}

\begin{lemma}
\label{lemma_c_2} From $\Psi(\cdot,Q)\equiv \xi(\cdot,P)$ follows $P|i(U)=P\neg i(U\backslash i)$ for all $\{i\}\subseteq U\subseteq N$ with
$|U|\ge 2$.
\end{lemma}

\begin{proof}
We set $U=N\backslash S\cup i$ so that we have to prove $P|i(N\backslash
S\cup i)=P|\neg i(N\backslash S)$ for all subsets $\{i\}\subseteq
S\subsetneq N$.

For fixed $S$ we consider the unanimity game $u_{N\backslash S}$ and obtain
the formulas
\begin{eqnarray*}
\xi_i(u_{N\backslash S},P) &=& \!\!\sum\limits_{T\ni i} u_{N\backslash
S}(T)\cdot P\!\!\mid\!\! i(T)- \sum\limits_{T\not\ni i} u_{N\backslash
S}(T)\cdot P\!\!\mid\!\! \neg i(T) \\
&=&\!\!\!\!\!\!\sum\limits_{T\colon N\backslash
S\subseteq T\subseteq N\backslash\{i\}}\!\!\!\!\!\!\!\!\!\!\!\!\!\!\! \Big(%
P\!\!\mid\!\! i(T\cup i)-P|\neg i(T)\Big)
\end{eqnarray*}
and
\begin{equation*}
\Psi_{i}(u_{N\backslash S},Q)=\sum\limits_{T\ni i}Q_i(T)\Big[u_{N\backslash S}(T)-u_{N\backslash S}(T\backslash i)\Big]%
=0.
\end{equation*}
Now we prove the proposed statement by induction on the subsets $S$ in
increasing order of their cardinalities using the assumption $%
\xi_i(u_S,P)=\Psi(u_S,Q)$. For the induction start $S=\{i\}$ we have $%
P|i(N)-P|\neg i(N\backslash i)=0$, which is equivalent to $P|i(N)=P|\neg
i(N\backslash\{i\})$. Using the induction hypothesis for all $S^{\prime
}\subseteq N$ with $|S^{\prime }|<|S|$ yields $P|i(N\backslash
S\cup i)=P|\neg i(N\backslash S)$.
\end{proof}

\medskip

Put $p_i:=\sum_{T\ni i
}P(T)\in[0,1]$ for all $i\in N$. Whenever $p_i>0$ we have $P|i(S)=\frac{P(S)%
}{p_i}$ for all $\{i\}\subseteq S\subseteq N$ and $P|i(S)=0$ in all other
cases. The next lemma excludes the case $p_i=1$ for at least two players.

\medskip

\begin{lemma}
\label{lemma_exclude_one} If $\Psi(\cdot,Q)\equiv \xi(\cdot,P)$ and if there
exists an index $i\in N$ with $p_i=1$, then $n=1$.
\end{lemma}

\begin{proof}
From $p_i=\sum_{T\ni i}P(T)=1$ we conclude $P(S)=0$ for all $S\subseteq N\backslash i $.
Thus we have $P|i(T)=0$ for all $T\ni i$ with $|T|\ge 2$ due to Lemma~\ref{lemma_c_2}. This yields $P(\{i\})=$ $Q_i(\{i\})=1$ and $Q_j(S)=0$ for all $S\ni j$, where $(S,j)\neq (\{i\},i)$, and all $j\in N$ due to Lemma~\ref{lemma_c_1}.
For each $j\in N\backslash i$ we then have $\sum_{S\ni j} Q_j(S)=0\neq 1$~--  a contradiction.
\end{proof}

\begin{proof}[Proof of Theorem \ref{thm_prob_val_eq_cpi}]
From Lemma~\ref{lemma_exclude_one} we conclude
$$0\le p_i:=\sum\limits_{T\ni i}P(T)<1$$
for all $i\in N$. If $%
p_i=0$ for an index $i\in N$, then we have $Q_i(S)=0$ due to Lemma~\ref%
{lemma_c_1}, which contradicts the definition of the $Q_i(S)$. Thus we have
$0<p_i<1$. Later on it will turn out that indeed we can choose $\tilde{p}%
_i=p_i$.

We have
\begin{equation*}
P(S)=\frac{p_i}{1-p_i}\cdot P(S\backslash i)
\end{equation*}
for all $S\ni i$ with $|S|\ge 2$ due to Lemma~\ref{lemma_c_2} and $p_i>0$.
Thus inductively we obtain
\begin{equation*}
P(S)=\prod_{j\in S\backslash i} \frac{p_j}{1-p_j}\cdot P(\{i\})
\end{equation*}
for all $i\in N$ and all subsets $S\ni i$ of $N$.

Inserting the previous equations into $p_i=\sum\limits_{S\ni i}P(S)$ yields
\begin{equation*}
p_i=P(\{i\})\cdot\sum\limits_{S\not\ni i} \prod_{j\in S}%
\frac{p_j}{1-p_j}=P(\{i\})\cdot\prod_{j\in N\backslash i} \left(\frac{p_j%
}{1-p_j}+1\right) =P(\{i\})\cdot\prod_{j\in N\backslash i}\frac{1}{1-p_j}.
\end{equation*}
Thus we have
\begin{equation*}
P(\{i\})=p_i\,\cdot\,\prod_{j\in N\backslash i}(1-p_j),
\end{equation*}
which then yields
\begin{equation}  \label{eq_product}
P(S)=\prod_{j\in S} p_j\,\cdot\,\prod_{j\in N\backslash S}(1-p_j)
\end{equation}
for all $\varnothing\neq S\subseteq N$. By using $\sum_{S\subseteq
N}P(S)=1$ we conclude that equation~(\ref{eq_product}) is also valid for the
empty set and thus for all subsets of $N$. \\
Lemma~\ref{lemma_c_1} and a short calculation gives also the first formula
of the proposed statement.\\ \vspace{0.5cm}

To verify that the converse holds as well let $0<p_i<1$ be given for all $%
i\in N$ and define
\[P(S)=\prod_{j\in S}p_j\cdot\prod_{j\in N\backslash S}
\left(1-p_j\right),
\] i.e., $P$ is a product measure. Next set
\[
Q_i(S\cup i)=P|i(S\cup i)=\prod_{j\in S}p_j\cdot\prod_{j\in
N\backslash (S\cup i)} \left(1-p_j\right),
\] for $i\in N\backslash S$, i.e.\ the $Q_i(S\cup i)$ derive
from the same product measure. We can easily verify $P|i(S)=P|\neg i(S%
\backslash i)$ for all $S\ni i$ and all $i\in N$.
Inserting this into the definition of the prediction value provides $\xi(\cdot,P)=\Psi(\cdot,Q)$.

\end{proof}

\bigskip

\setlength{\labelsep}{-0.265cm}

\newcommand{\noopsort}[1]{}

\end{document}